\documentclass[11pt]{article}

\usepackage{amssymb,amsmath,algorithm,algorithmic,enumerate,
verbatim,bm,fullpage,amsthm,enumerate,thmtools}
\usepackage[usenames,dvipsnames,svgnames,table]{xcolor}
\usepackage[pagebackref,letterpaper=true,colorlinks=true,pdfpagemode=none,citecolor=OliveGreen,linkcolor=BrickRed,urlcolor=BrickRed,pdfstartview=FitH]{hyperref}

\newcommand\R{{\mathbb R}}

\def\B{\{0,1\}}

\def\bx{{\bf x}}
\def\by{{\bf y}}
\def\bz{{\bf z}}

\newcommand\norm[1]{\|#1\|}
\newcommand\I[1]{\mathbb{I} \left[{#1}\right]}
\renewcommand\P[1]{\mathbb{P}\left[{#1}\right]}
\newcommand\PP[2]{\mathbb{P}_{#1}\left[{#2}\right]}

\declaretheorem[numberwithin=section]{theorem}
\declaretheorem[sibling=theorem]{lemma}
\declaretheorem[sibling=theorem]{proposition}
\declaretheorem[sibling=theorem]{claim}

\declaretheorem[sibling=theorem]{remark}
\declaretheorem[sibling=theorem]{fact}
\declaretheorem[sibling=theorem]{definition}

\newenvironment{proofof}[1]{{\medbreak\noindent \em Proof of #1.  }}{\hfill\qed\medbreak}

\newcommand\E{\mathop{\mathbb E}\displaylimits}
\newcommand\pr{\mathop{\mathbb P}\displaylimits}
\DeclareMathOperator{\poly}{poly}
\DeclareMathOperator{\sse}{SSE}
\DeclareMathOperator{\diam}{diam}
\DeclareMathOperator{\supp}{supp}
\DeclareMathOperator{\argmin}{argmin}

\DeclareMathOperator{\sdp}{sdp}
\DeclareMathOperator{\sdpL}{L(k)}
\DeclareMathOperator{\sdpSA}{SA(2k)}

\def\to{\rightarrow}
\def\eps{\epsilon}
\def\cR{{\mathcal R}}
\def\d{d}
\def\cD{{\mathcal D}}
\def\cP{{\mathcal P}}
\def\hd{{d}}
\def\sd{d^2}
\def\radius{{\delta}}
\def\L{{\mathcal L}}

\newcommand\cE[1]{{\mathcal E}(#1)}

\begin{document}

\title{Improved ARV Rounding in Small-set Expanders \\and Graphs of Bounded Threshold Rank}

\author{Shayan Oveis Gharan
\thanks{Department of Management Science and Engineering, Stanford University. Supported by a Stanford Graduate Fellowship. My research is partly supported by  grant \#FA9550-12-1-0411
from the U.S. Air Force Office of Scientific Research (AFOSR) and the
Defense Advanced Research Projects Agency (DARPA).
Email:\protect\url{shayan@stanford.edu}.}
\and
Luca Trevisan\thanks{Department of Computer Science, Stanford University. This material is based upon  work supported by the National Science Foundation under grant No.  CCF 1017403 and CCF 1216642 and United States - Israel Binational Science Foundation under grant No. BSF 2010451.
Email:\protect\url{trevisan@stanford.edu}.}
}

\maketitle

\begin{abstract}
We prove a structure theorem for the feasible solutions of the Arora-Rao-Vazirani SDP relaxation
on low threshold rank graphs and on small-set expanders.
 We show that if $G$ is a graph of bounded threshold rank or a small-set expander, then
an optimal solution of the Arora-Rao-Vazirani relaxation (or of any stronger version of it) can be almost entirely
covered by a small number of balls of bounded radius. 

Then, we show that, if $k$ is the number of balls, a solution of this form can be rounded with an approximation factor of $O(\sqrt {\log k})$ in the case of the Arora-Rao-Vazirani relaxation, and with a constant-factor approximation in the
case of the $k$-th round of the Sherali-Adams hierarchy starting at the Arora-Rao-Vazirani relaxation.

The structure theorem and the rounding scheme combine to prove the following result, where $G=(V,E)$ is
a graph of expansion $\phi(G)$, $\lambda_k$ is the $k$-th smallest eigenvalue of the normalized Laplacian of $G$, and $\phi_k(G) = \min_{\text{disjoint } S_1,\ldots,S_k} \max_{1\leq i\leq k} \phi(S_i)$ is the largest expansion of any $k$ disjoint subsets of $V$:
 if either  $\lambda_k \gtrsim \log^{2.5} k \cdot \phi(G)$ or $\phi_{k} (G) \gtrsim \log k \cdot \sqrt{\log n}\cdot \log\log n\cdot \phi(G)$, then the Arora-Rao-Vazirani relaxation can be rounded in polynomial time with an approximation ratio $O(\sqrt{\log k})$.

Stronger approximation guarantees are achievable in time exponential in k via relaxations in the Lasserre hierarchy. Guruswami and Sinop \cite{GS13} and Arora, Ge and Sinop \cite{AGS13} prove that $1+\eps$ approximation is achievable in time
$2^{O(k)}\poly(n)$ if either $\lambda_k > \phi(G)/\poly(\eps)$, or if $\sse_{n/k} > \sqrt{\log k \log n} \cdot \phi(G)/\poly(\eps)$,
where $\sse_s$ is the minimal expansion of sets of size at most~$s$.

\end{abstract}

\section{Introduction}
We study approximation algorithms for the uniform sparsest cut problem in regular
graphs\footnote{The analysis can be extended to the case of general graph, in which
case the problem becomes the {\em conductance} problem, and the condition on small-set
expansion becomes a condition on small-set conductance.} based on semidefinite programming.
Let $G=(V,E)$ be a $r$-regular graph. The expansion of a set $S\subseteq V$ is the ratio 
$$\phi(S) := \frac{|E(S,\bar S)|}{r\cdot |S|}.$$
We want to find  the smallest
expansion of nonempty subsets of size at most $n/2$.
 We use $\phi(G)$ to denote the
value of the optimum,
$$\phi(G):= \min_{S: 0< |S| \leq n/2} \ \phi(S).$$  
We will not define the uniform sparsest cut problem, but any approximation of expansion is,
up to an additional factor of 2, an approximation of the uniform sparsest cut problem. 

Chawla et al.~\cite{CKKRS06} show that assuming unique games conjecture there
is no constant factor approximation algorithm for a more general version of the problem known as non-uniform sparsest cut problem.
The best polynomial time (or even sub-exponential time) approximation algorithm for uniform sparsest cut
problem remains the algorithm of Arora, Rao and Vazirani \cite{ARV04}, which achieves an
$O(\sqrt{\log  n})$ approximation factor, where $n$ is the number of vertices. The Arora-Rao-Vazirani
algorithm is based on a semidefinitite programming (SDP) relaxation, which we will
refer to as the ``ARV'' relaxation.

Within the long-term research program of developing better approximation algorithms for uniform sparsest cut in general graphs, there has been much success in the past few years toward developing better algorithms for restricted classes of graphs.

The technique of subspace enumeration \cite{KT07,Kolla11,ABS10} applies to the special class of graphs known as ``low threshold rank'' graphs. 
Let $G=(V,E)$ be a $r$-regular graph, and  $\L:=I-A/r$ be the normalized laplacian matrix of $G$, where $I$ is the identity matrix and $A$ is the adjacency matrix. Let 
$$ 0=\lambda_1\leq \lambda_2\leq \ldots\lambda_n\leq 2,$$
be the eigenvalues of $\L$. A graph $G$ has a low threshold rank, if $\lambda_k=\Omega(1)$ for a small number $k$. Low threshold rank graphs can be considered as a generalization of expander graphs.

Arora, Barak and Steurer \cite{ABS10} show that  the technique of {\em subspace enumeration} developed in the work of  Kolla and Tulsiani~\cite{KT07,Kolla11},  achieves
a $O(1/\lambda_k)$ approximation in time $2^{O(k)}\poly(n)$.
Later, Barak, Raghavendra and Steurer \cite{BRS11} and Guruswami and Sinop \cite{GS11} match this $O(1/\lambda_k)$ approximation factor in time $n^{O(k)}$ by using an SDP relaxation that is derived from the Arora-Rao-Vazirani
relaxation by $k$ ``rounds'' of a procedure defined by Lasserre. The procedure starts from
a SDP relaxation of a combinatorial problem which can be formulated as a 0/1 integral program,
and defines a family of relaxations with additional variables and constraints. The $k$-th
relaxation in this family has size $n^{O(k)}$. We will refer to the $k$-th relaxation in this family
as $L_k$-ARV.

These techniques are very powerful, and they lead to approximation algorithms for many constraint satisfaction problems including maxcut, sparsest cut, min uncut, graph coloring etc \cite{AG11,GS12,GS13,OT12}.
These algorithms run in time that is exponential in $k$, and they typicallyprovide an approximation ratio of $1/\poly(\lambda_k)$.
A notable exception is the work of Guruswami and Sinop \cite{GS13}, who show that even if $\lambda_k \ll 1$, but there is a gap between $\lambda_k$ and $\phi(G)$, such as, say, $\lambda_k > 2 \phi(G)$,  then a constant-factor approximation
can be derived from $L_k$-ARV. The approximation ratio can be made arbitrarily close to one if the ratio $\lambda_k/\phi(G)$is sufficiently large.
Because of the exponential dependency on $k$ in the running time of all of the above algorithms,  we may obtain a polynomial time algorithm
only for graphs with a {\em fast growing spectrum}, i.e., if $k=O(\log n)$.

Let $\sse_s(G)$ denote the minimal expansion of a subset of $\leq s$ vertices of $G$,
$$\sse_s(G):=\min_{0<|S|\leq s} \phi(S).$$ 
Arora, Ge and Sinop \cite{AGS13} have recently announced a new rounding scheme for $L_k$-ARV that
achieves a $(1+\eps)$ approximation provided that $\sse_{n/k} (G) \gtrsim  \sqrt{\log k\cdot \log n} \cdot \sdpL / \poly(\eps)$, where $\sdpL$ is the optimum solution of the $k$ rounds of Lasserre hierarchy applied to the ARV relaxation. Roughly speaking, provided that there is a large enough gap (of order $\sqrt{\log n\log k}$) between the expansion in sets of size $\leq n/2$ and in sets of size $\leq n/k$, they can design a PTAS for the uniform sparsest cut problem. Qualitatively, this
restriction is in the same spirit as a restriction on graphs in which $\lambda_k$ is large. Indeed,
\cite{LOT12,LRTV12} prove that $\sse_{n/k} \leq O(\sqrt{\lambda_{2k} \log k})$, so that
a requirement on small-set expansion is a stronger assumption than a requirement on $\lambda_k$.
Quantitatively, however, the quadratic gap in the results of \cite{LOT12,LRTV12} makes the
restriction on small-set expansion much weaker in some cases. For example, a cycle
has expansion $\approx 1/n$, and $\sse_{n/k} \approx k/n$, while $\lambda_k \approx \frac{k^2}{n^2}$. This means that the condition $\lambda_k > 2 \phi$ is not satisfied for $k=o(\sqrt n)$, while
the condition $\sse_{n/k} > (\log k)^{O(1)} \cdot \log n \cdot \phi$ is satisfied for
$k = \tilde O(\log n)$.
Similar to the low threshold rank graphs, \cite{AGS13} may obtain a polynomial time
algorithm only if $\sse_s(G) > \poly(\log n) \cdot \phi(G)$ for $s=\Omega(n/\log n)$.

\subsection{Our Results}
In this paper we design improved polynomial time approximation algorithms for uniform sparsest cut problem on
graphs with a {\em moderately} growing spectrum, i.e., when the threshold rank can be as large as $n^{o(1)}$. 
We show that if there is a $\poly(\log k)$ gap between $\lambda_k$ and the optimum of ARV relaxation, then  the solution of ARV relaxation can be rounded with an approximation factor of $O(\sqrt{\log k})$ in polynomial time.
\begin{theorem}
\label{thm:ranklogk}
There are  universal constants $c_1,c_3$ and a randomized rounding algorithm such that for any graph $G$, if $\lambda_k > \frac{c_1}{c_3} \cdot \log^{2.5} k \cdot \sdp$, 
then the Arora-Rao-Vazirani relaxation can be rounded with an approximation ratio $O(\sqrt{\log k})$.
Here, $\sdp$ is the optimum value of the ARV relaxation.
\end{theorem}
For example, if $k=n^{o(1)}$, we get a $o(\sqrt{\log n})$ approximation for uniform sparsest cut.
Consequently, to improve the $O(\sqrt{\log n})$ approximation of \cite{ARV04}
it is sufficient to find an improved approximation algorithm for the high threshold rank graphs. 

We can also improve the ARV rounding if the graph is a small set expander, i.e. $\sse_{n/k}(G) \gg \phi(G)$, and in fact a weaker condition, based on a parameter that we define next, suffices.  The {\em order $k$ expansion} of $G$ is the smallest number $\phi_k(G)$ such that
we can find $k$ disjoint subsets of vertices each of expansion at most $\phi_k(G)$,
$$ \phi_k(G):=\min_{\text{disjoint } S_1,\ldots,S_k\subseteq V} \max_{1\leq i\leq k} \phi(S_i).$$ 

Note that for any graph $G$, $\phi_k(G) \geq \sse_{n/k} (G)$, and $\phi_k(G)\geq \lambda_k/2$ but 
$\lambda_k$ and $\sse_{n/k}(G)$ are incomparable.
Consequently, having a lower-bound on $\phi_k(G)$ is weaker than having lower bounds
on any of $\lambda_k$ or $\sse_{n/k}(G)$\footnote{For example, if $G$ is a complete graph with a separated vertex, then $\phi_k(G)=\Theta(1)$ for any $k\geq 2$ whereas $\sse_s(G)=0$ for any $s\geq 1$.}. We show that if there is a $\tilde{O}(\log k \sqrt{\log n})$ gap between $\phi_k(G)$, and the optimum of ARV relaxation, then the solution of ARV relaxation can be rounded with an approximation factor of $O(\sqrt{\log k})$ in polynomial time.
\begin{theorem}
\label{thm:sselogk}
There are universal constant $c_2,c_3$ and a randomized rounding algorithm such that
for any graph $G$, if $\phi_k(G) > \frac{c_2}{c_3}\cdot  \log k\cdot \sqrt{\log n}\cdot \log\log n\cdot \sdp$, 
then the Arora-Rao-Vazirani relaxation can be rounded with an approximation ratio $O(\sqrt{\log k})$.
Here, $\sdp$ denotes the optimum value of ARV relaxation.
\end{theorem}
Note that \autoref{thm:sselogk} is not necessarily stronger than \autoref{thm:ranklogk} because
the gap between $\phi_k(G)$ and $\sdp$ is a function of $n$ that can be significantly larger than $k$.

In addition to the above results we can obtain a constant factor approximation algorithm
if  we are given a feasible solution of the $k$-th round of the Sherali-Adams hierarchy 
starting at the Arora-Rao-Vazirani relaxation. 
\begin{theorem}
\label{thm:ranksseconst}
There are universal constants $c_1,c_2,c_4$ and a randomized rounding algorithm
such that for any graph $G$, if $\lambda_k(G) > \frac{c_1}{c_4}\cdot \log^2 k \cdot \sdpSA$
or if $\phi_k(G) > \frac{c_2}{c_4}\cdot \sqrt{\log k \log n}\cdot \log\log n\cdot \sdpSA$, then
  the $2k$-th round of the Sherali-Adams hierarchy
of the Arora-Rao-Vazirani relaxation can be rounded with a constant factor approximation.
Here, $\sdpSA$ denotes the optimum value of the $2k$ rounds of Sherali-Adams hierarchy applied to ARV relaxation. 
\end{theorem}

The above result is weaker than results of Arora, Ge, Guruswami and Sinop. In \cite{GS13}, Guruswami and Sinop achieve an arbitrarily good approximation, instead of a constant factor approximation, assuming that the ratio between $\lambda_k$ and the optimum of the relaxation is a sufficiently large constant, while we need the ratio to be of the order of $\log^2 k$. In \cite{AGS13}, Arora, Ge and Sinop achieve an arbitrarily good approximation, instead of a constant factor approximation, assuming that the ratio between $\sse_{n/k}(G)$ and the optimum of the relaxation is at least order
of $\sqrt{\log k \log n}$, while we need it to be order of $\sqrt{\log k \log n}\log\log n$. On the other hand, our results hold for the weaker Sherali-Adams relaxation, are proved via rather different techniques and, in the case of the second result, are based on an assumption on $\phi_k(G)$ rather than $\sse_{n/k}(G)$.


%
%

\subsection{Structure Theorems and Improved Rounding}
A feasible
solution of the ARV relaxation for a graph $G=(V,E)$ is an assignment of a vector $\bx_v$
to every vertex $v\in V$. The vectors are normalized so that $\sum_{u,v\in V}  \norm{\bx_u - \bx_v}^2 = n^2$ and the
cost of the solution is
\[ \sdp:= \frac 1{2r\cdot n} \cdot {\sum_{(u,v)\in E} \norm{ \bx_u-\bx_v}^2} \]
where $r$ is the degree. (Furthermore, the distance function $d_\bx(u,v) := \norm{\bx_u-\bx_v}^2$ satisfies
the triangle inequality.) 

Our proofs follow from structure theorems on feasible solutions of ARV relaxation. We show that if $\lambda_k$ or $\phi_k$ are sufficiently larger than $\sdp$, then the SDP solution  
can be almost entirely covered by a small number (at most $2k$) of balls of small radius. 
Then, we use this property to design  improved rounding algorithms for graphs with moderately growing spectrum or small set expanders. 

Our structure theorems are in the spirit of on a recent work of Kwok et al. \cite{KLLOT13}.
They show that the nearly-linear time ``sweep'' algorithm based on the eigenvector of the second smallest
eigenvalue of $\L$ has approximation ratio $O(k/\sqrt {\lambda_k})$. Their paper
is based on a structure theorem showing that the eigenvector of the second eigenvalue
of the Laplacian must be at distance at most $O(\sqrt{\lambda_2/\lambda_k})$ from a vector
whose entries contain only $2k$ distinct values.

We show a similar structure theorem for the solutions of ARV relaxation.
Our structure theorems show that there are universal constants $c_1,c_2$ such that for any $0<\radius <1$, if $\lambda_k \geq c_1\cdot  \log^2k \cdot \sdp/\radius$ or if $\phi_k(G) \geq c_2 \cdot \sqrt{\log n\log k}\cdot \log\log n\cdot  \sdp /\radius$, then 
we can find $2k$ sets $T_1,\ldots,T_{2k}$ of diameter 
(according to the distance $\d_\bx(u,v) := \norm{\bx_u - \bx_v}^2$) at most $\radius$ covering
almost all vertices of $G$. 

If the solution of ARV relaxation has exactly $2k$ points, then we can give a $O(\sqrt{\log k})$
approximation using ARV rounding, or a constant factor approximation using $2k$-levels of Sherali-Adams hierarchy ($2k$-levels of Sherali-Adams hierarchy enforces that on
every $2k$ vertices we have an integral cut). Our improved rounding algorithms provide a robust version of these facts.

Our improved rounding scheme for ARV shows that if we can cover almost all of the vertices 
by $2k$ sets of diameter $\radius$, then provided that $\radius \leq c_3/\sqrt{\log k}$, where $c_3$ is a universal constant, 
we can find in polynomial time a cut of expansion at most $O(\sqrt{\log k} \cdot \sdp)$.
This proves \autoref{thm:ranklogk} and \autoref{thm:sselogk}.

If the solution $\{\bx_v \}_{v\in V}$ as above is feasible for $2k$ rounds of Sherali-Adams applied to ARV,
then, provided $\radius < c_4$, where $c_4 >0$ is a universal constant, we can find in polynomial time a solution of cost at most $O(\sdp)$. This proves \autoref{thm:ranksseconst}

\subsection{Techniques}
\subsubsection{The Structure Theorems}

To prove our structure theorems, our general idea is to divide the ambient space of the vectors $\bx_v$ into regions of small
diameters, and to look at the $2k$ regions with the most elements of the well-spread set guaranteed by
the argument of Leighton and Rao. If such regions cover almost all the vertices in the well-spread set, then
the covered vertices are themselves a well-spread set, and we are done. Otherwise, we find either $k$
disjointly-supported functions $f_i : V \to \R$ all of small Rayleigh quotient, leading to a contradiction to
the assumption on $\lambda_k$, or $k$ disjoint sets each of small expansion, leading to a contradiction
to the assumption on $\phi_k$.

To implement this idea in the case in which the assumption is on $\lambda_k$, we use, as in \cite{LOT12}, a
``padded decomposition'' of the ambient space of the vector. The problem with low-diameter padded decomposition, 
however, is that, if $h$ is the dimension of the space, we will lose a factor of $h^2$ in the Rayleigh quotient (see \autoref{prop:partitioninglambda}).
The ARV solution $\{ \bx_v \}_{v\in V}$ that we start from may lie in a $n$-dimensional space, and so we would need
an assumption of the form $\sdp >  n^2 \cdot \lambda_k$ in order to prove our result. To avoid this loss,
we first map our solution $\{ \bx_v \}_{v\in V}$ to an $O(\log k)$-dimensional space. The resulting solution
$\{ \bz_v \}_{v\in V}$ may not satisfy the triangle inequality any more, but it has approximately the same cost
as the solution $\{ \bx_v \}_{v\in V}$, and for all but a $o(1/k^2)$ fraction of the pairs $u,v$ the
distances $\norm{\bx_u - \bx_v}^2$ and $\norm{\bz_u - \bz_v}^2$ are approximately the same. (But, we emphasize again,
the distances $\norm{\bz_u-\bz_v}^2$ may not satisfy the triangle inequality.) Using a padded low-diameter
decomposition, and the assumption $\lambda_k > c_1 \cdot \log^2 k\cdot \sdp/\radius$, we are then able to cover almost
all of the vertices using $2k$ regions that have diameter at most $\radius/2$ according to the distances
$\norm{\bz_u - \bz_v}^2$. We can then argue that we can also cover almost all the points using $2k$ sets
whose diameter is at most $\radius$ according to $\norm{\bx_u - \bx_v}^2$, or else we have $\Omega(n^2/k^2)$ pairs
$(u,v)$ such that the dimension-reduction distorted their distance by a constant factor, which happens with low
probability.

To cover a well-spread set under the assumption on $\phi_k$ we have to overcome an additional problem:
each edge will be cut with probability proportional to $\norm{\bz_u-\bz_v} \approx \norm{\bx_u - \bx_v}$ in the low-diameter decomposition, but
it contributes only $\norm{\bx_u - \bx_v}^2$ to the cost $\sdp$ of the solution $\{ \bx_v \}_{v\in V}$, so that the fraction of cut
edges (and consequently the assumption on $\phi_k$) would be proportional on $\sqrt{\sdp}$ rather than $\sdp$.
We resolve this difficulty by first using a result of Arora, Lee and Naor \cite{ALN05} to map the solution $\{ \bx_v \}$ to vectors $\{ \by_v \}_{v\in V}$ such that for {\em all} pairs $(u,v)$ we have $\norm{ \bx_u -\bx_v}^2 \leq \norm{\by_u-\by_v} \leq \tilde O(\sqrt {\log n})
\cdot \norm{\bx_u - \bx_v}^2$. Then we map to a solution $\{ \bz_v\}_{v\in V}$ in a $O(\log k)$-dimensional space such that
for all but $o(n^2/k^2)$ pairs $u,v$ we have $\norm{\by_u-\by_v} \leq 2 \norm{\bz_u-\bz_v}$. Using this solution, we
 use the assumption $\phi_k > c_2 \cdot \sqrt {\log n\log k} \cdot \log\log n\cdot \sdp/\radius$ to find $2k$
regions of diameter $\radius/2$, according to the distances $\norm{\bz_u-\bz_v}$ that cover almost all the vertices. Reasoning as before, this means that we also have $2k$
sets of diameter $\radius$, according to the distances $\norm{\by_u-\by_v}$ that cover almost all the vertices. Since $\norm{\bx_u-\bx_v}^2\leq \norm{\by_u-\by_v}$ for all pairs $u,v\in V$, the diameter of the $2k$ sets is at most $\radius$ according to the distances $\norm{\bx_u-\bx_v}^2$. 

\subsubsection{The Rounding Schemes}

An argument that
goes back to Leighton and Rao \cite{LR99} shows that either a simple rounding algorithm succeeds in finding
a cut of expansion at most $4\sdp$, or else there is a set $A\subseteq V$ that is $(4,1/16)$-{\em well-spread}.

\begin{definition} Let $\{ \bx_v \}_{v\in V}$ be an assignment of vectors to vertices such 
that $\sum_{u,v} \norm{\bx_u-\bx_v}^2 = n^2$. We say that a set $A\subseteq V$ is $(\alpha,\beta)$-well-spread if for all $u,v\in A$,
\begin{eqnarray*}
 \norm{\bx_u - \bx_v}^2 &\leq& \alpha, \\
 \sum_{u,v \in A} \norm{\bx_u - \bx_v}^2 &\geq& \beta \cdot n^2. 
\end{eqnarray*}
\end{definition}


Let $\{\bx_v\}_{v\in V}$ be  a feasible solution to the ARV relaxation, and let $A$ be a $(O(1),\Omega(1))$-well-spread
set of vertices. 
For a set $U\subseteq V$, and $v\in V$, let $\d_\bx(v,U):=\min_{u\in U} \norm{\bx_u-\bx_v}^2$.
If we can find a set  $U \subseteq A$ such that 
$\sum_{u,v} |\d_\bx(u,U)-\d_\bx(v,U)| \geq n^2/\radius$,
then we can find a set of expansion at most $O(\sdp/\radius)$.
Arora, Rao and Vazirani \cite{ARV04} show that such set can always be found, with $\radius= \Omega(1/\sqrt {\log n})$.

Using a refinement of the result of \cite{ARV04} proved by Arora, Lee and Naor \cite{ALN05}, we can show
that if the solution is such that the well-spread set $\{\bx_v\}_{v\in A}$ occupies only $k$ distinct points, then we can have $\radius=\Omega(1/\sqrt {\log k})$,
and that this remains true if $\{\bx_v\}_{v\in A}$ can be covered by $k$ sets of diameter $o(1/\sqrt{\log k})$.

Our other rounding scheme assumes that $\{\bx_v\}_{v\in V}$ is (part of) a solution that is feasible
for $2k$ rounds of Sherali-Adams applied to ARV. If a well-spread set $A$ is covered by $2k$ low-diameter
sets $A_1,\ldots,A_{2k}$, let $C=\{ a_1,\ldots,a_{2k}\}$ be a set of $2k$ vertices, with each $a_i\in A_i$ chosen 
arbitrarily from each region. The feasibility for $2k$ rounds of Sherali Adams implies that there is a
probability distribution over metrics $\cD(\cdot,\cdot)$ (described as part of the feasible solution) such that for every pair $(u,v)$ we 
have $\norm{\bx_u-\bx_v}^2 = \E[\cD(u,v)]$, and we also have with probability 1 that $\sum_{u,v} \cD(u,v) = \sum_{u,v} \norm{\bx_u-\bx_v}^2 = n^2$ and that there is a partition $C_L,C_R$ of $C$ such that $\cD(a_i,a_j)$ is zero if $a_i,a_j$  are both in $C_L$
or both in $C_R$, and it is the same positive value for all pairs in $C_L \times C_R$. (That is, with probability 1, $\cD(\cdot,\cdot)$ is a cut metric on $C$.) Overall, sampling from this distribution, we find, with constant probability, a
distance function $\cD$ such that
\begin{eqnarray*}
\sum_{(u,v)\in E} \cD(u,v) &\leq& O(1) \cdot \sum_{(u,v)\in E} \norm{\bx_u - \bx_v}^2 \\
\sum_{u,v \in A} \cD(u,v) &\geq& \Omega( n^2)\\
\sum_{u\in A} \cD(u,C) &\leq& O(\delta\cdot n)
\end{eqnarray*}
which is enough, for small enough constant $\delta$, to  round the solution and find
a cut of expansion at most $O(\sdp)$.

\section{Preliminaries}
\subsection{SDP Relaxations}
The Arora-Rao-Vazirani relaxation is defined as follows; the variables $\bx_v$, one
for each vertex, are vectors.

\begin{equation}
\begin{aligned}
\mbox{min \ \ \ \ \ \ \ \ } & \frac1{2r\cdot n}\sum_{(u,v)\in E} \norm{\bx_u - \bx_v}^2\\
\mbox{subject to\ }
& \sum_{u,v\in V} \norm{\bx_u - \bx_v}^2 = n^2\\
& \norm{\bx_u - \bx_v}^2 \leq \norm{\bx_u - \bx_w}^2 + \norm{\bx_w - \bx_v}^2 & \forall u,v,w\in V 
\end{aligned}
\label{eq:arv}
\end{equation}

\bigskip
Next, we show that this SDP is indeed a relaxation for the uniforms sparsest cut problem.
Let $S\subseteq V$ with size $s:=|S|$ be the optimum solution, i.e., $0<s\leq n/2$ and $\phi(S)=\phi(G)$.
Then, for every $v\in V$, we define
\begin{equation}
\label{eq:integralsdp}
\bx_v=
\begin{cases}
\sqrt{\frac{n^2}{s(n-s)}} & \text{if } v\in S,\\
0 & \text{otherwise.}
\end{cases}
\end{equation}
Since the above assignment defines a cut, it satisfies the triangle inequality. Furthermore, since $s\leq n/2$,
$$ \frac{1}{2r\cdot n} \sum_{(u,v)\in E} \norm{\bx_u-\bx_v}^2 = \frac{1}{2r\cdot n}|E(S,\overline{S})| \frac{n^2}{s(n-s)} \leq \phi(S).$$
Therefore, \eqref{eq:arv} is a relaxation of the uniform sparest cut problem.

The following relaxation is a simplification of the $k$-th round Sherali-Adams strengthening of ARV.
There is a vector variable $\bx_v$ for every vertex $v\in V$, then there are also $2^k \cdot {n \choose k}\cdot {n\choose 2}$ additional real
variables $\d^{R,b}(u,v)$, one for every two vertices $u,v\in V$, every subset $R\subseteq V$ of
cardinaility $k$, and every map $b: R \to \B$, and finally there are $2^k\cdot {n\choose k}$ real variables
$p^{R,b}$. 

\begin{equation}
\begin{aligned}
\mbox{min  \ \ \  \ \ \ \ \ } & \frac{1}{2r\cdot n}\sum_{(u,v)\in E} \norm{\bx_u - \bx_v}^2\\
\mbox{subject to }
& \sum_{u,v\in V} \norm{\bx_u - \bx_v}^2 = n^2\\
& \norm{\bx_u - \bx_v}^2 \leq \norm{\bx_u - \bx_w}^2 + \norm{\bx_w - \bx_v}^2 & \forall u,v,w\in V \\
\\
& \sum_{u,v\in V} \d^{R,b} (u,v) = n^2\cdot p^{R,b} & \forall R,b\\
&\d^{R,b} (u,v) \leq \d^{R,b}(u,w) + \d^{R,b} (w,v) & \forall u,v,w\in V \\
\\
& \norm{\bx_u-\bx_v}^2 = \sum_{b:R\rightarrow \B}\d^{R,b} (u,v) & \forall u,v\in V \forall R\\
& \sum_{b:R\rightarrow \B} p^{R,b} = 1 & \forall R\\
& p^{R,b} \geq 0 & \forall R,b\\
& \d^{R,b} (u,v) = 0 & \forall u,v \forall R,b. \ b(v)=b(u)\\
\end{aligned}
\label{eq:arvk}
\end{equation}

The relaxation \eqref{eq:arvk} can be interpreted in the following way. For every subset $R\subseteq V$ of $k$
vertices, we have a probability distribution $p^{R,b}$ over the possible cuts $b:R \to \{0,1\}$ of $R$. 
For each such $R$ and $b$, we have a distance function $\d^{R,b}$ that satisfies the triangle inequality,
and which has the following property. Fix any $R$, and consider the following distribution on metrics: pick a random $b: R \to \{0,1\}$ with
probability $p^{R,b}$, and define ${\mathcal D}(u,v):= \frac {1}{p^{R,b}} \d^{R,b} (u,v)$. Then each $\cD(\cdot,\cdot)$ in this
sample space satisfies the triangle inequality and, for all pair of vertices $u,v\in V$,
\begin{eqnarray}
\label{eq:expdistance}
\E[\cD(u,v)] &=& \norm{\bx_u-\bx_v}^2,\\
\label{eq:sumdistance}
\sum_{u,v} \cD(u,v) &=& n^2.
\end{eqnarray}
Furthermore, $\cD(\cdot,\cdot)$
defines a cut metric over $R$.

Next, we show that \eqref{eq:arvk} is a relaxation of the uniform sparsest cut problem.
Let $S$ be the optimum solution. Then, we assign the same value as in \eqref{eq:integralsdp} to each vector $\bx_v$. For each set $R\subset V$, let $b_R: R\rightarrow \{0,1\}$ be the indicator of $S$ in $R$, that is for any $v\in R$,
$$ b_R(v):=\begin{cases}
1 & \text{if } v\in S,\\
0 & \text{otherwise}.
\end{cases}$$
Then, for any $b:R\rightarrow \{0,1\}$, we let $p^{R,b} := \I{b=b_R}$. This defines a probability
distribution for each set $R$.
Furthermore, for any $u,v\in V$ let $d^{R,b}(u,v) := \I{b=b_R} \norm{\bx_u-\bx_v}^2$.
It is easy to see this satisfies all of the constraints of \eqref{eq:arvk}.
\subsection{Notations}
Let  $\{\bx_v\}_{v\in V}\in \mathbb{R}^m$ be be a sequence of vectors 
assigned to the vertices of $G$. In this paper we will use two types 
of distance functions on these vectors. For any pair of vertices $u,v\in V$, we use
\begin{eqnarray*} 
\sd_{\bx} (u,v) &:=& \norm{\bx_u-\bx_v}^2\\
\hd_{\bx}(u,v) &:=& \norm{\bx_u-\bx_v}.
\end{eqnarray*}
Note that $\hd_\bx$ is just the Euclidean distance, and is always a metric, but $\sd_\bx$ is not necessarily a metric. In fact,
if $\sd_\bx$ is a metric, we say vectors $\{\bx_v\}_{v\in V}$ form a {\em negative type} metric.

For a distance function $d(.,.)$, we define the energy of the graph 
as follows:
\begin{eqnarray*} 
{\mathcal E} (d) &:=& \frac{1}{2r\cdot n}\sum_{(u,v)\in E} \d(u,v),  \\
\end{eqnarray*}
For a distance function $d(.,.)$, and a set $S\subseteq V$, the diameter of $S$ is the maximum distance between the vertices of $S$,
$$ \diam(S,d):=\max_{u,v\in S} d(u,v).$$
For every $u\in V$, let
$$ \d(u,S) := \min_{v\in S} \d(v,u).$$
For any $u\in V$ and $r>0$, the ball of radius $r$ about $u$ is the set of vertices at distance at most $r$ from $u$,
$$ B_\d(u,r):=\{v: \d(u,v)\leq r\}$$

\subsection{Spectral Graph Theory}

For a $r$-regular graph $G=(V,E)$, the Rayleigh quotient of a function $f:V\rightarrow \mathbb{R}$, is defined
as follows,
$$ \cR(f):=\frac{\sum_{(u,v)\in E} (f(u)-f(v))^2}{r\cdot \sum_{v\in V} f^2(v)}.$$
A proof of the following statement can be found in \cite[Lem 2.3]{KLLOT13}.
\begin{fact}
\label{fact:lambda}
For any graph $G$, and any $k$ disjointly supported functions $f_1,\ldots,f_k: V\rightarrow \mathbb{R}$,
$$\lambda_k \leq 2\max_{1\leq i\leq k} \cR(f_i).$$
\end{fact}

\subsection{Random partitions of metric spaces}
\label{sec:rps}

We now discuss some of the theory of random partitions
of metric spaces.
Let $d(.,.)$ be a metric defined on the vertices of $G$. 
We will write a partition $P$ of $V$ as a
function $P : V \to 2^V$ mapping a vertex $v \in V$
to the unique set in $P$ that contains $v$.

For $\radius > 0$, we say that $P$ is {\em $\radius$-bounded} if
$\diam(S,d) \leq \radius$ for every $S \in P$.
We will also consider distributions over random partitions.
If $\mathcal P$ is a random partition of $V$, we say that $\mathcal P$
is $\radius$-bounded if this property holds with probability one.

A random partition $\mathcal P$ is {\em $(\radius, \alpha, \eps)$-padded} if
$\mathcal P$ is $\radius$-bounded, and
for every $v \in V$, we have
$$
\pr[B_d(v, \radius/\alpha) \subseteq \mathcal P(v)] \geq \eps.
$$
A random partition is {\em $(\radius,L)$-Lipschitz} if $\mathcal P$ is $\radius$-bounded,
and, for every pair $u,v \in V$, we have
$$
\pr[\mathcal P(u) \neq \mathcal P(v)] \leq L \cdot \frac{d(u,v)}{\radius}\,.
$$

Here are some results that we will need.
The first theorem is known, more generally, for doubling spaces \cite{GKL03},
but here we only need its application to $\mathbb R^k$.
See also \cite[Lem 3.11]{LN05}.

\begin{theorem}\label{thm:rkpad}
If vertices are mapped to $\mathbb{R}^k$, and $d(.,.)$ is the Euclidean distance between the vertices, then for every $\radius > 0$ and $\eps > 0$, $(V,d)$ admits a $(\radius, O(k/\eps), 1-\eps)$-padded random partition.
\end{theorem}

The next result is proved in \cite{CCGGP98}.
See also \cite[Lem 3.16]{LN05}.

\begin{theorem}\label{thm:RkLip}
If vertices are mapped to $ \mathbb R^k$, and $d(.,.)$ is the Euclidean distance between the vertices, then for every $\radius > 0$, $(V,d)$ admits a $(\radius, O(\sqrt{k}))$-Lipschitz random partition.
\end{theorem}

\section{Statement of Our Results}

\begin{theorem} [$\lambda_k$ Structure Theorem] 
\label{thm:structlambda}
There is a universal constant $c_1$ such that for any sequence of vectors $\{\bx_v\}_{v\in V}$,
 and $\eps,\radius>0$, if
 \[ \lambda_k > c_1 \cdot \log^2 k \cdot \cE{\sd_\bx} \cdot \frac {\log^{2} (1/\eps)}{\eps^3 \cdot \radius}, \]
then there are $2k$ sets $T_1,\ldots,T_{2k}$ with diameter $\diam(T_i,\sd_\bx)\leq \radius$ covering $1-\eps$ fraction of vertices. Furthermore, we can find these sets by a randomized polynomial time algorithm.
\end{theorem}

\begin{theorem}[$\phi_k$ Structure Theorem]
\label{thm:structphi}
There is a universal constant $c_2$ such that for any sequence of vectors $\{\bx_v\}_{v\in V}$ that form a metric of negative type,
and $\eps,\radius>0$, if
\[ \phi_k > c_2 \cdot \sqrt{\log k \cdot \log n} \cdot \log\log n \cdot \cE{\sd_\bx} \cdot \frac {\sqrt{\log(1/\eps)}}{\eps\cdot \radius}, \]
then, there are $2k$ sets $T_1,\ldots,T_{2k}$ with diameter $\diam(T_i,\sd_\bx)\leq \radius$ covering $1-\eps$ fraction of vertices. Furthermore, we can find these sets by a randomized polynomial time algorithm.
\end{theorem}

\begin{theorem}[Rounding ARV] 
\label{thm:arvrounding}There are universal constants $c_3,\eps_1>0$ and a randomized polynomial time rounding algorithm such that the following holds.

For any feasible solution to \eqref{eq:arv}, $\{  \bx_v \}_{v\in V}$ of cost $\sdp$,
if there are $2k$ sets $T_1,\ldots,T_{2k}$ of diameter $\leq c_3/\sqrt{\log k}$ covering $(1-\eps_1)n$ vertices, then the rounding algorithm finds a cut of expansion at most ${\sdp} \cdot O(\sqrt {\log k})$ with high probability.
\end{theorem}

\begin{theorem}[Rounding Sherali-Adams Relaxations] \label{thm:sa} There are universal constants $c_4,\eps_1>0$ and a randomized polynomial time rounding algorithm such that the following holds.

For any feasible solution to \eqref{eq:arvk},  $\{  \bx_v \}_{v\in V}$ of cost $\rm sdp$,
if there are $2k$ sets $T_1,\ldots, T_{2k}$ of diameter $\leq c_4$ covering $(1-\eps_1)n$ vertices, then
 the rounding algorithm finds a cut of expansion at most $O(\sdp)$ with high probability.
 \end{theorem}

\section{The Structure Theorems}
\subsection{$\lambda_k$ structure theorem}

In this section we prove \autoref{thm:structlambda}.
Using the following lemma we project the solution of SDP to a $O(\log k)$ dimension space.
This helps us to decrease the loss in partitioning to a function that is only depend (poly-logarithmically) on $k$. We show that the solution of SDP only suffer by a constant factor and most of the pair of vertices will not get distorted.
\begin{lemma}
\label{lem:dimreduction}
For any sequence of vectors $\{ \bx_v \}_{v\in V} \subseteq \mathbb{R}^m$, and  $\eps>0$,
there exists vectors $\{ \bz_v\}_{v\in V} \in \mathbb{R}^h$ where $h = \Theta( \log 1/\eps)$ such that
\begin{eqnarray*}
\cE{\hd_\bz} &\leq & 4\cE{\hd_\bx},\\
\cE{\sd_\bz} &\leq & 4\cE{\sd_\bx},\\
\big| \{ (u,v): \norm{\bx_u-\bx_v}^2 &>& 2 \norm{\bz_u-\bz_v}^2 \} \big| \leq  \eps n^2.
\end{eqnarray*}
\end{lemma}
\begin{proof}
Let $g_1, g_2, \ldots, g_h$ be i.i.d. $m$-dimensional Gaussians, and
consider the random mapping $\Gamma_{m,h} : \mathbb R^k \to \mathbb R^h$ defined by
$\Gamma_{m,h}(x) = h^{-1/2} (\langle g_1, x \rangle, \langle g_2, x \rangle, \ldots, \langle g_h, x \rangle)$.
Then we have the following basic estimates (see, e.g. \cite[Ch. 15]{Mat01} or \cite[Ch. 1]{LT11}).
For every $x \in \mathbb R^m$,
\begin{eqnarray}
\label{eq:guassl22}
\E\left[\|\Gamma_{m,h}(x)\|^2\right] &=& \|x\|^2,\\
\label{eq:guassl2}
\E\left[ \norm{\Gamma_{m,h}(x)} \right] &\leq& \norm{x}.
\end{eqnarray}
Equation \eqref{eq:guassl2} follows from \eqref{eq:guassl22} and the Jensen's inequality. 
Also for every $\alpha \in (0,\frac12]$,
\begin{equation}\label{eq:guasscon}
\pr\left[\vphantom{\bigoplus} (1-\alpha)\norm{x}^2\leq  \norm{\Gamma_{m,h}(x)}^2 \leq   (1+\alpha)\norm{x}^2\right] \geq 1-2 e^{-\alpha^2 h/12}\,,
\end{equation}
Choose $h=\Theta(\log 1/\eps)$ such that $2e^{-h/48} < \eps/4$. 
Then, by \eqref{eq:guasscon} for each $u,v\in V$, we get
$$ \pr\left[ \norm{\Gamma_{m,h}(\bx_v - \bx_u)}^2 < \norm{\bx_v-\bx_u}^2/2 \right] \leq \eps/4.$$
By linearity of the $\Gamma_{m,h}$ operator, 
$$ \Gamma_{m,h}(\bx_v-\bx_u)  = \Gamma_{m,h}(\bx_v) - \Gamma_{m,h}(\bx_u).$$
Therefore, by \eqref{eq:guassl22} and \eqref{eq:guassl2} and Markov's inequality with probability 1/4 we get vectors $\{ \Gamma_{m,h}(\bx_v) \}_{v\in V}\subseteq \mathbb{R}^h$
that satisfies all inequalities in lemma's statement.
\end{proof}

The following proposition is the main part of the proof. Here we show that if
we can not find $2k$ regions of small diameter covering almost all of the vectors $\{\bz_v\}_{v\in V}$ then $\lambda_k$ must be very large.
\begin{proposition}
\label{prop:partitioninglambda}
There is a universal constant $c_1'$ such that for any sequence of vectors $\{\bz_v\}_{v\in V}\in\mathbb{R}^h$, 
 and $0<\eps,\radius<1/2$
if 
$$ \lambda_k \geq c'_1\cdot \frac{h^2}{\eps^3\cdot \radius^2} \cE{\sd_\bz}, $$
then there are $2k$ sets $S_1,\ldots,S_{2k}$  such that $\diam(S_i,\hd_{\bz})\leq \radius$ for $1\leq i\leq 2k$, and
$$ \Big|\bigcup_{i=1}^{2k} S_i\Big| \geq (1-\eps) n.$$
\end{proposition}
\begin{proof}
We show that if we can not find $2k$ sets of small diameter covering almost all of the vertices, then we can construct $k$ disjointly supported function of Rayleigh quotient $O(h^2 \cE{\sd_\bz} /\eps^3 \radius^2)$. Then by \autoref{fact:lambda} we get a contradiction
to the assumption that $\lambda_k$ is large. The proof has three main steps:
random partitioning, merging into dense regions, and construction of disjointly supported functions.

{\bf i) Random partitioning.} 
By \autoref{thm:rkpad} there exists a $(\radius, \alpha, 1-\eps/4)$-padded random
partitioning for the metric space $(V,\hd_\bz)$, for $\alpha\lesssim h/\eps$. 
For $S\subseteq V$, the {\em interior} of $S$ is the set
$$ \tilde{S}:=\{v: B_{\hd_\bz}(v,\radius/\alpha)\subseteq S\}.$$
By linearity of expectations and Markov inequality, with probability at least $1/2$ we sample a partition $P$ such that, for 
each $S\in P$, $\diam(S,\hd_\bz) \leq \radius$, and
\begin{equation}
\label{eq:lambdacovering}
 \sum_{S\in P} |\tilde{S}| \geq (1-\eps/2)n.
 \end{equation}
Choose a labeling $S_1,S_2,\ldots$ of the sets of $P$ such that
\begin{equation}
\label{eq:lambdaordering}
 |\tilde{S}_1| \geq |\tilde{S}_2| \geq \ldots \geq |\tilde{S}_{|P|}|.
 \end{equation}

If $\sum_{i=1}^{2k} |\tilde{S}_i| \geq (1-\eps)n$, then  $S_1,\ldots,S_{2k}$ form $2k$ sets of diameter $\radius$ covering $(1-\eps)n$ 
vertices, and we are done. Otherwise, we get a contradiction to the assumption that $\lambda_k$ is large.
Using equation \eqref{eq:lambdacovering}, for the rest of the proof we assume
\begin{equation}
\label{eq:lambdacase}
\sum_{i>2k} |\tilde{S}_i| \geq \eps n/2.
\end{equation}

Next, we merge the sets in $P$ greedily into $2k$ disjoint sets $T_1,\ldots,T_{2k}$ such that
$\min_{1\leq i\leq 2k} |\tilde{T}_i|\geq \frac{\eps n}{8k}$. Then, we construct 2k disjointly supported functions $f_1,\ldots,f_{2k}$ where $\supp(f_i) \subseteq T_i$. 
Finally, We show that the best $k$ of these $2k$ functions have Rayleigh quotient $O(\alpha^2 \cE{\sd_\bz}/\eps\radius^2)$.

{\bf ii) Merging.} We start by constructing $T_1,\ldots,T_{2k}$. First let $T_i=S_i$, for $1\leq i\leq 2k$. Then, iteratively, for each $i>2k$ we merge $S_i$ with the set with smallest interior, $\argmin_{T_j} |\tilde{T}_j|$. 
In the next claim we show that, by the end of the algorithm, all $T_i$ have at least
$\eps n/8k$ vertices in their interior.
\begin{claim}
\label{cl:merging}
$$\min_{1\leq i\leq 2k}|\tilde{T}_i| \geq \eps n/8k.$$
\end{claim}
\begin{proof}
If $|\tilde{S}_{2k}|\geq \eps n/8k$, then we are done, since, by equation \eqref{eq:lambdaordering}, for all $1\leq i\leq 2k$,
$$ |\tilde{T}_i|\geq |\tilde{S}_i| \geq |\tilde{S}_{2k}|.$$
Otherwise,  we have $\tilde{S}_{2k+1} < \eps n/8$. By equation \eqref{eq:lambdacase},
$$ \sum_{i=1}^{2k} |\tilde{T}_i \setminus \tilde{S}_i | \geq \eps n/2.$$
Therefore, for some $1\leq j\leq 2k$, $|\tilde{T}_j\setminus \tilde{S}_j| \geq \eps n/4$.
Let $S_l$ be the last set that is merged with $T_j$. 
Since $S_l$ is merged with $T_j$, $T_j$ is  the set with the smallest interior 
 at the time of merging with $S_l$. Therefore, 
 $$ \min_{1\leq i\leq 2k} |\tilde{T}_i| \geq |\tilde{T}_j \setminus \tilde{S}_l| \geq |\tilde{T}_j \setminus \tilde{S}_j \setminus \tilde{S}_l| \geq \frac{\eps n}{4} - \frac{\eps n }{8} = \frac{\eps n}{4}.$$
\end{proof}
{\bf iii) Construction of disjointly supported functions.} It remains to construct $k$ disjointly supported functions of small Rayleigh quotient.
For each $1 \leq i\leq 2k$, let $f_i: V\rightarrow \mathbb{R}$ be defined as
$$ f_i(v):=\max(0, 1-\alpha \cdot \hd_\bz(v,\tilde{T}_i)/\radius).$$
Note that for $v\in \tilde{T}_i$, $f_i(v)=1$,  and the value of $f_i(v)$ decreases smoothly with the distance of $v$ from $\tilde{T}_i$; in particular, $f_i(v)=0$, for $v\notin T_i$.
It follows from the above definition that each $f_i$ is $\alpha/\radius$-lipschitz w.r.t. $\hd_\bz(.,.)$, that is, for all $u,v\in V$ and $1\leq i\leq 2k$,
$$  |f_i(u)-f_i(v)| \leq \frac{\alpha}{\radius} \hd_\bz(u,v).$$
Since $T_1,\ldots,T_{2k}$ are disjoint, $f_1,\ldots,f_{2k}$ are disjointly supported.
Therefore,
$$ \sum_{i=1}^{2k} \sum_{(u,v)\in E} (f_i(u)-f_i(v))^2 
= \frac{2\alpha^2}{\delta^2}\sum_{(u,v)\in E} \sd_\bz(u,v).$$
By Markov inequality, after relabeling, we get $k$ disjointly supported functions $f_1,\ldots,f_k$ such
that for all $1\leq i\leq k$,
$$ \sum_{(u,v)\in E} (f_i(u)-f_i(v))^2 \leq \frac{2\alpha^2}{\delta^2 k} \sum_{(u,v)\in E} \sd_\bz(u,v).$$
By an application of \autoref{fact:lambda} we get
\begin{eqnarray*} \lambda_k \leq 2\max_{1\leq i\leq k} \cR(f_i) = 2\max_{1\leq i\leq k} \frac{\sum_{(u,v)\in E} (f_i(u)-f_i(v))^2}{r\cdot \sum_v f_i^2(v)} &\leq&
\frac{\frac{4\alpha^2}{\radius^2\cdot k} \sum_{(u,v)\in E} \sd_\bz(u,v)}{  r\cdot \eps n /8k}\\
& =& \frac{64\alpha^2\cE{\sd_\bz}}{\eps\cdot \radius^2} \lesssim \frac{h^2\cE{\sd_\bz}}{\eps^3\cdot \radius^2 }
\end{eqnarray*}
Therefore,  $\lambda_k \leq c''_1 h^2 \cE{\sd_\bz} / \eps^3\radius^2$ for a universal constant $c''_1$. Letting $c_1' > c_1''$ we get a contradiction.
\end{proof}

The following lemma is the last step of the proof. We show that sets of small diameter in $O(\log k)$ dimension space
corresponds to sets of small diameter in the original space of the ARV solution.
\begin{lemma}
\label{lem:projectback}
Let $d(.,.), d'(.,.)$ be two distance functions on $V$.
Let $S_1,\ldots,S_{2k}\subseteq V$ be
such that $\diam(S_i,d)\leq \radius$ for all $1\leq i\leq 2k$, and let $S:=\cup_{i=1}^{2k} S_i$. 
If, for $\eps>0$,
$$|\{(u,v): d'(u,v) > 2d(u,v)\}| \leq \frac{\eps^3n^2}{64 k^2},$$ 
then, there are $2k$ sets $T_1,\ldots,T_{2k}$ of diameter $4\radius$ w.r.t. to $d'$ that cover all but $\eps n$ of the vertices of $S$.
\end{lemma}
\begin{proof}
Wlog we assume that $S_1,\ldots,S_{2k}$ are disjoint (we can simply delete multiple copies of a single vertex).
We say a set $S_i$ is {\em good}, if there exists $u\in S_i$ such that
$$ |B_{\d'}(u,2\radius) \cap S_i| \geq (1-\eps/2)|S_i|,$$
and it is {\em bad} otherwise. We use $A$ to denote the indices of good sets. By the above definition, for each good set $S_i$, there is a ball $B_{d'}(u,2\radius)$  such that $|B_{d'}(u,2\radius)\cap S_i|\geq (1-\eps/2)|S_i|$. 
We define $T_i$ to be such ball if $i\in A$ and $T_i=\emptyset$ otherwise. Then,
$$ \Big|\bigcup_{i=1}^{2k} T_i\Big| \geq \sum_{i=1}^{2k} |T_i \cap S_i| \geq \sum_{i\in A} |S_i|(1-\eps/2).$$
Since $|T|\leq n$, to prove the lemma it is sufficient to show that
\begin{equation}
\label{eq:goodsets} \sum_{i \in A} |S_i| \geq |T|-\eps n/2. 
\end{equation}
Equivalently, we can show $\sum_{i\notin A} |S_i|\leq \eps n/2$.
We say that a pair of vertices $u,v$ is {\em distorted}, if $d'(u,v)>2d(u,v)$. Observe
that each bad set $S_i$ contains $\eps |S_i|^2/4$ distorted pairs. This is because 
 for each vertex $u\in S_i$ we have $S_i \subseteq B_{\d}(u,\radius)$, but at least $\eps |S_i|/2$ of the vertices of $S_i$
are not contained in the ball $B_{\d'}(u,2\radius)$.

Since $S_1,\ldots,S_{2k}$ are disjoint, the total number of distorted pairs is at least,
$$ \sum_{i\notin A} \frac{\eps |S_i|^2}{4} \geq \frac{\eps}{4} \Big( \frac{\sum_{i\notin A} |S_i|}{2k-|A|}\Big)^2 \geq \frac{\eps}{16k^2} \left(\sum_{i\notin A} |S_i|\right)^2.$$
On the other hand, by the lemma's assumption, the number of distorted pairs is at most 
$\eps^3 n^2/64k^2$. This implies that $\sum_{i\notin A} |S_i| \leq \eps n/2$ which
proves \eqref{eq:goodsets} and 
$$ \sum_{i=1}^{2k} |T_i| \geq |S|-\eps n.$$
\end{proof}

\begin{proofof}{\autoref{thm:structlambda}}
First, by \autoref{lem:dimreduction}, there are vectors $\{\bz_v\}_{v\in V}\subseteq \mathbb{R}^h$,
where $h=\Theta( \log k \log 1/\eps)$,
such that $\cE{\sd_\bz}\leq 4\cE{\sd_\bx}$, and
\begin{equation}
\label{eq:distortionlambda}
\big| \{ (u,v): \sd_\bx(u,v) > 2 \sd_\bz(u,v) \} \big| \leq  \frac{\eps^3 n^2}{512 k^2 }.
\end{equation}

Let $c_1:=\frac{64 c_1'  h^2}{(\log k\log 1/\eps)^2} $. Then, since $\cE{\sd_\bz}\leq 4\cE{\sd_\bx}$, the assumption of the theorem implies 
$$ \lambda_k \geq c_1' \cdot \frac{h^2}{(\eps/2)^3 (\sqrt{\radius/4})^2} \cE{\sd_\bz}.$$
By \autoref{prop:partitioninglambda}, there are $2k$ sets $S_1,\ldots,S_{2k}$
such that $\diam(S_i,\sd_\bz) = \left(\diam(S_i,\hd_\bz)\right)^2 \leq \radius/4$ for $1\leq i\leq 2k$, and they cover $(1-\eps/2)n$ vertices of $G$. 
Finally, by \eqref{eq:distortionlambda}, we can apply \autoref{lem:projectback} with $\d = \sd_\bz, \d' = \sd_\bx$, we get $2k$ sets
$T_1,\ldots,T_{2k}$ with diameter $\diam(T_i,\sd_\bx)\leq \radius$ that cover 
$$(1-\eps/2) \Big| \bigcup_{i=1}^{2k} S_i\Big| \geq (1-\eps/2) (1-\eps/2)n \geq (1-\eps)n$$ 
vertices of $G$. Observe that all steps of the analysis are constructive, so we have a polynomial time randomized algorithm for finding $T_1,\ldots,T_{2k}$.
\end{proofof}

\subsection{$\phi_k$ structure theorem} 

The proof of the next proposition is very similar to the proof of  \autoref{prop:partitioninglambda}.
The main difference is that we use Lipschitz partitioning instead of padded partitioning. This is because we just want to find $k$ disjoint non-expanding sets, so we just need that very few edges are cut by the random partitioning. 
\begin{proposition}
\label{prop:partitioningphi}
There is a universal constant $c'_2$ such that for any sequence of vectors $\{\bz_v\}_{v\in V}\in \mathbb{R}^h$, and $0 < \eps,\radius<1/2$, if
$$ \phi_k \geq c'_2 \cdot \frac{\sqrt{h}}{\eps \cdot \radius} \cE{\hd_\bz}, $$
then there are $2k$ sets $T_1,\ldots,T_{2k}$ such that $\diam(T_i,\hd_\bz)\leq \radius$
for $1\leq i\leq 2k$, and
$$ \Big| \bigcup_{i=1}^{2k} T_i\Big| \geq (1-\eps)n.$$
\end{proposition}
\begin{proof}
We show that if we cannot find $2k$ regions of small diameter covering almost all vertices
then we can construct $k$ disjoints sets of small expansion.

{\bf i) Random partitioning.}
By \autoref{thm:RkLip} $(V,\hd_\bz)$ admits a $(\radius,\alpha)$-Lipschitz random partition
$\cP$ for $\alpha=O(\sqrt{h})$ such that for all pair of vertices $u,v$
$$ \P{\cP(u) \neq \cP(v)} \leq \alpha \cdot \frac{\hd_\bz(u,v)}{\radius}.$$
By linearity of expectation and Markov inequality, with probability at least $1/2$ we sample  a partition $P$ such that for all $S\in P$, $\diam(S,\hd_\bz)\leq \radius$, and
\begin{equation}
\label{eq:cutedges}
|\left\{(u,v)\in E: P(u) \neq P(v)\right\}| \leq \frac{2\alpha}{\radius} \sum_{(u,v)\in E} \hd_\bz(u,v).
 \end{equation}
Choose a labeling $S_1,S_2,\ldots$ of the sets in $P$ such that
\begin{equation}
\label{eq:phiordering}
 |S_1| \geq |S_2| \geq \ldots \geq |S_{|P|}|.
 \end{equation}
If $\sum_{i=1}^{2k} |S_i| \geq (1-\eps) n$, then $S_1,\ldots, S_{2k}$ are $2k$ sets
of diameter $\radius$ covering $(1-\eps)n$ vertices and we are done.
Otherwise, we get a contradiction to the assumption that $\phi_k$ is large.
For the rest of the proof we assume
\begin{equation}
\label{eq:sums}
 \sum_{i>2k} |S_i| \geq \eps n.
 \end{equation}
Next, we merge the sets into $2k$ disjoint sets $T_1,\ldots,T_{2k}$ such that 
each $T_i$ covers at least $\eps n/4k$ vertices of $G$. Then, we show that 
at least $k$ of these sets have expansion $O(\alpha \cE{\hd_\bz}/\eps\radius)$.

{\bf ii) Merging.}
This part is very similar to the Merging part of \autoref{prop:partitioninglambda}. We start by constructing $T_1,\ldots,T_{2k}$. First let $T_i=S_i$ for $1\leq i\leq 2k$. Then, iteratively, for each $i>2k$ merge $S_i$
with $\argmin_{T_i} |T_i|$. 
We show that by
the end of this procedure 
$$
\min_{1\leq i\leq 2k} |T_i|\geq \frac{\eps n}{4k}.
$$
In particular,  if $|S_{2k+1}|\geq \eps n/4k$, then by \eqref{eq:phiordering}
we get the above equation.
Otherwise,
by equation \eqref{eq:sums}, 
there is a set $T_j$ such that $|T_j\setminus S_j|\geq \eps n/2k$. If $S_l$ is the
last set merged with $T_j$, by the time $S_l$ was merging with $T_j$, $T_j$ has the least number of vertices. Therefore,
 $$\min_{1\leq i\leq 2k} |T_i| \leq |T_j\setminus S_l| \geq |T_j\setminus S_j \setminus S_l| \geq \eps n/4k.$$

{\bf iii) Construction of non-expanding sets.} By equation \eqref{eq:cutedges}
$$ \sum_{i=1}^{2k} |E(T_i,\overline{T_i})| 
\leq |\left\{ (u,v)\in E: P(u)\neq P(v)\right\}|
\leq \frac{2\alpha}{\radius} \sum_{(u,v)\in E} \hd_\bz(u,v).$$
Therefore, by Markov inequality, after relabeling, we get $k$ disjoint sets $T_1,\ldots,T_k$
such that for all $1\leq i\leq k$,
$$ \phi(T_i) = \frac{|E(T_i,\overline{T_i})|}{d\cdot |T_i|} \leq \frac{\frac{2\alpha}{\radius\cdot k} \sum_{(u,v)\in E} \hd_\bz(u,v)}{ d\cdot \eps n/4k} = \frac{16\alpha}{\eps\cdot \radius} \cE{\hd_\bz} \lesssim \frac{\sqrt{h}}{\eps\cdot\radius}\cE{\hd_\bz}.$$
Therefore, $\phi_k \leq c''_2 \cdot \sqrt{h} \cE{\hd_\bz} / \eps \radius$ where $c''_2$ is a universal constant. Letting $c'_2>c''_2$ we get a contradiction. 
\end{proof}

\begin{proofof}{\autoref{thm:structphi}}
Since $\{\bx_v\}_{v\in V}$ form a negative type metric, by \cite{ALN05}, there
are vectors $\{\by_v\}_{v\in V}$ such that for all $u,v\in V$,
\begin{equation}
\label{eq:aln}
\norm{\bx_v-\bx_u}^2 \leq \norm{\by_v-\by_u} \leq \alpha\norm{\bx_v-\bx_u}^2.
\end{equation}
where $\alpha\lesssim \sqrt{\log n}\cdot \log\log n$.
By above equation $\cE{\hd_\by}\leq \alpha \cE{\sd_\bx}$.
By \autoref{lem:dimreduction} we can reduce the dimension to $O(\log k)$. There are vectors $\{\bz_v\}_{v\in V}\subseteq \mathbb{R}^h$
for $h = \Theta( \log k \log 1/\eps)$
such that $\cE{\hd_\bz}\leq 4\cE{\hd_\by}$, and
\begin{equation}
\label{eq:distortionphi}
\big| \{ (u,v): \hd_\by(u,v) > 2 \hd_\bz(u,v) \} \big| \leq  \frac{\eps^3 n^2}{512 k^2 }.
\end{equation}

Let $c_2:= 32 c_2'  \cdot \alpha \cdot \sqrt{\frac{h}{\log k\log 1/\eps}} $. Then, since $\cE{\hd_\bz}\leq 4\alpha\cE{\sd_\bx}$ by theorem's assumption we get
$$ \lambda_k \geq c_2' \cdot \frac{\sqrt{h}}{(\eps/2) (\radius/4)} \cE{\hd_\bz}.$$
By \autoref{prop:partitioningphi}, there are $2k$ sets $S_1,\ldots,S_{2k}$
such that $\diam(S_i,\hd_\bz) \leq \radius/4$ for $1\leq i\leq 2k$, and they cover $(1-\eps/2)n$ vertices of $G$. 
By \eqref{eq:distortionphi}, we can apply \autoref{lem:projectback} with $\d = \hd_\bz, \d' =\hd_\by$, we get $2k$ sets
$T_1,\ldots,T_{2k}$ with diameter $\diam(T_i,\hd_y)\leq\radius$ that cover
$$(1-\eps/2) \Big| \bigcup_{i=1}^{2k} S_i\Big| \geq (1-\eps/2) (1-\eps/2)n \geq (1-\eps)n$$ 
vertices of $G$. Finally using \eqref{eq:aln} we get that
the dimeter of each set $T_i$ is at most $\radius$ w.r.t. $\sd_\bx$,
$$ \diam(T_i,\sd_\bx) \leq \diam(T_i,\hd_\by) \leq \delta.$$
Observe that all steps of the analysis are constructive, so we have a polynomial time randomized algorithm for finding $T_1,\ldots,T_{2k}$.
\end{proofof}

%

\section{The Rounding Algorithms}
Linial, London, Robinovich \cite{LLR95} observed that by defining a {\em Frechet embedding} we can round any metric $d(.,.)$ into a set of small expansion.
\begin{fact}
\label{fact:frechet}
For any metric $d(.,.)$ that satisfies the triangle inequality, and any set $U\subseteq V$,
there is a polynomial time algorithm that finds a set of expansion
$$ \frac{\sum_{(u,v)\in E} |d(u,U)-d(v,U)|}{\frac rn\cdot \sum_{u,v} |d(u,U)-d(v,U)|}$$
\end{fact}

Leighton, Rao~\cite{LR99} and Arora, Rao and Vazirani \cite{ARV04} used the above simple fact to show the following preprocessing step.

\begin{fact}[\cite{ARV04}] 
\label{fact:wellspread}
There is a polynomial time algorithm that given a feasible solution  $\{ x_v \}_{v\in V}$ to the ARV relaxation \eqref{eq:arv} of cost $\sdp$ either finds a cut of expansion $O(\sdp)$, or finds a set $W$ that is $(4,1/16)$-well spread.
\end{fact}

Observe that if $W$ is a $(4,1/16)$-well spread set as in the conclusion of  \autoref{fact:wellspread}, and $A$ is obtained from $W$ by removing $n/128$ vertices
from $W$, then $A$ is still $(4,1/32)$-well spread. Let $\eps_1:=1/128$.
Therefore, given $2k$ disjoint sets $T_1,\ldots,T_{2k}$ covering
$(1-\eps_1)n$ vertices of $G$, and a $(4,1/16)$-well spread set $W$ as in the conclusion of \autoref{fact:wellspread}, the set $A:=W\cap (\cup_{i=1}^{2k} T_i)$ is $(4,1/32)$-well spread. 
Furthermore, we can write $A$ as $A=A_1\cup A_2\cup \ldots\cup A_{2k}$ 
where $A_i:=W\cup T_i$.

So, in the rest of this section we assume $A$ is $(4,1/32)$ well spread and {\em all} of its vertices
are covered by $2k$ sets $A_1,\ldots,A_{2k}$ of diameter $\radius$ according to $\d_{\bx}$.


\subsection{Rounding ARV}

Given a feasible solution to \eqref{eq:arv} that has a well-spread subset of vertices, by \autoref{fact:frechet} finding a good rounding
reduces to finding a set $U\subseteq A$ that is well-separated from the rest of the vertices of $A$.

\begin{fact} 
\label{fact:sepset}
Let $\{\bx_v\}_{v\in V}$ be a feasible solution to \eqref{eq:arv} of cost $\sdp$, and suppose that $A\subseteq V$
is a $(O(1),\Omega(1))$-well spread set. Then there is a polynomial time algorithm that given
$\{\bx_v \}_{v\in V}$ and a subset $U \subseteq A$, finds a cut of expansion at most 

\[ O(1) \cdot \frac {n^2 \cdot \sdp} { \sum_{u,v\in A} |\d_{\bx}(v,U)-\sd_\bx(u,U)|}. 
\]
\end{fact}

Arora et al.~\cite{ARV04} prove that  a set $U$ such that $\sum_{u,v\in A} |\sd_\bx(v,U)-\sd_\bx(u,U)| \geq \Omega(n^2/\sqrt{\log n})$ always exists and can be found in polynomial time. This is best possible in general. 

Suppose, however, that $\{ \bx_v \}_{v\in A}$ concentrate in just $k$ distinct points. Then let $C\subseteq A$ be a set
of ``representatives,'' that is $|C|=k$ and we have that for every $v\in A$ there is a unique $u\in C$ such that $\bx_v=\bx_u$.
Let also $w(u)$ be the number of vertices $v\in A$ such that $\bx_v=\bx_u$. The condition that $A$ is well-spread
can be written as

\[ \sum_{u,v\in C} w(u) \cdot w(v)\cdot \I{\sd_\bx(u,v) \geq \Omega(1)} \geq \Omega(n^2) \]
and we are looking for a subset $U\subseteq C$ such that $\sum_{u,v\in C} w(u)w(v) |\sd_\bx(v,U)-\sd_\bx(u,U)| \geq \Omega(n^2/\sqrt{\log k}).$

\begin{lemma}[ARV with weights]
\label{lem:arvweighted}
Let  $\{\bx_u\}_{u\in C}$ be a sequence of vectors assigned to vertices of $C$ that form a metric of 
negative type.
For any set of weights $w: C\rightarrow \mathbb{R}_+$,
if 
$$\sum_{u,v} w(u) w(v) \I{ \sd_\bx(u,v) \geq \Delta} = \alpha \cdot n^2,$$
then there is a set $U\subseteq C$ such that 
$$\sum_{u,v} w(u)w(v) |\sd_\bx(u,U)-\sd_\bx(v,U)| \gtrsim \alpha \cdot \Delta \cdot n^2/\sqrt{\log|C|}.$$
\end{lemma}
\begin{proof}
We use the following lemma from Arora, Lee, Naor
\begin{lemma}[Arora et al.~\cite{ALN05}]
There exist constants $c_5 \geq 1$ and $0 < p_1 < 1/2$ such that for sequence of vectors $\{\bx_v\}_{v\in C}$ that form a negative type metric, and every $\Delta > 0$, the following holds. There exists a distribution $\mu$ over subsets $U \subset C$ such that for every $u, v \in C$ with $\sd_\bx(u, v) \geq \Delta$,
$$\PP{\mu}{u\in U \text{ and } \sd_\bx(v,U)\geq \frac{\Delta}{c_5\sqrt{\log|X|}}}\geq p_1.$$
\end{lemma}

We show that with a constant probability a random set $U$ satisfies the lemma. 
Define 
$$X_{u,v} :=\I{ \sd_\bx(u,v) \geq \Delta \text{ and } u\in U \text{ and } \sd_\bx(v,U)\geq \frac{\Delta}{c_5\sqrt{\log|X|}}}.$$
Then, by above lemma,
$$ \E\left[ \sum_{u,v} w(u) w(v) X_{u,v}\right] = p_1 \sum_{u,v} w(u)w(v) \I{d(u,v)\geq \Delta} = p_1 \cdot \alpha \cdot n^2.$$
Since $\sum_{u,v} w(u) w(v) X_{u,v}\leq \alpha \cdot n^2$ with probability 1, we get
$$\P{ \sum_{u,v} w(u) w(v) X_{u,v} \geq p_1 \cdot \alpha \cdot n^2} \geq p_1.$$
Therefore, with probability $p_1$ we get a set $U$ such that
$$ p_1 \cdot \alpha\cdot n^2\leq \sum_{u,v} w(u) w(v) X_{u,v} \leq \sum_{u,v} w(u) w(v) \frac{|\sd_\bx(u,U)-\sd_\bx(v,U)|}{\Delta/c_5\sqrt{\log|X|}} $$
\end{proof}

\begin{proofof}{\autoref{thm:arvrounding}}
Wlog we assume that $A$ is a $(4,1/32)$ well spread and all of its vertices are covered
by $2k$ {\em disjoint} sets $A_1,\ldots,A_{2k}$ of diameter $\delta\leq c_3/\sqrt{\log k}$ according to $\sd_\bx$.
For $\alpha=1/256, \Delta=1/64$, we have
\begin{equation}
\label{eq:farpairs}
\sum_{u,v\in A} \I{\sd_\bx(u,v)\geq \Delta} \geq \alpha\cdot n^2.
\end{equation}

Choose a vertex $a_i$ as a center from each set $A_i$, and let $C:=\{ a_1,\ldots,a_{2k}\}$.
Also, let $w(a_i)=|A_i|$. Let $c: A \rightarrow C$, where for each vertex $v\in A$, if $v\in A_i$, then $c(a)=a_i$. 
Then, for $\delta < \Delta/4$,
\begin{equation}
 \sum_{u,v\in A} \I{\sd_\bx(u,v)\geq \Delta} \leq \sum_{u,v\in A} \I{\sd_\bx(c(u), c(v)) \geq \Delta/2} = \sum_{u,v\in C} w(u) w(v) \I{\sd_\bx(u,v)\geq \Delta/2}.
 \label{eq:farcenters}
 \end{equation}
where in the first inequality we used the fact that $\sd_\bx$ is a metric. Putting \eqref{eq:farpairs} and \eqref{eq:farcenters} together we can apply \autoref{lem:arvweighted} and we get a set $U\subseteq A$ such that
\begin{equation}
\label{eq:avgpairwisedistance}
 \sum_{u,v\in C} w(u) w(v) |\sd_\bx(u,U) - \sd_\bx(v,U)| \geq \frac{c'_3 \cdot  n^2}{\sqrt{\log k}}, 
 \end{equation}
where $c'_3>0$ is a universal constant. Now choose $c_3 =\min(c'_3/2\sqrt{\log k}, \Delta/4)$. We get,
\begin{eqnarray*} 
\sum_{u,v\in A} |\sd_\bx(u,U)-\sd_\bx(v,U)| &\geq& \sum_{u,v\in A} |\sd_\bx(c(u),U) - \sd_\bx(c(v),U)| - |A| \sum_{v\in A} \sd_\bx(v,c(v))\\
&\geq& \sum_{u,v\in C} |\sd_\bx(u,U) - \sd_\bx(v,U)| - \delta n^2 \gtrsim \frac{n^2}{\sqrt{\log k}}.
\end{eqnarray*}
where the last inequality follows by \eqref{eq:avgpairwisedistance}. The theorem follows by an application of \autoref{fact:sepset}.
\end{proofof}

\subsection{Rounding Sherali-Adams Relaxations}

\begin{proofof}{\autoref{thm:sa}} 
Wlog we assume that $A$ is a $(4,1/32)$ well spread and all of its vertices are covered
by $2k$ {\em disjoint} sets $A_1,\ldots,A_{2k}$ of diameter $\delta\leq c_4$ according to $\sd_\bx$.
We choose a vertex $a_i$ arbitrarily as the center of $A_i$, and define $C:=\{a_1,\ldots,a_{2k}\}$
Also, let $c:A\rightarrow C$, where for  each $v\in A$, if $v\in A_i$, then $c(a)=a_i$.

Let $R=C$, and pick a random $b: R \to \{0,1\}$ with
probability $p^{R,b}$, and define ${\mathcal D}(u,v):= \frac {1}{p^{R,b}} \d^{R,b} (u,v)$. 
Then, we get \eqref{eq:expdistance} and \eqref{eq:sumdistance}, i.e., $\sum_{u,v} \cD(u,v) = n^2$, and for all pairs $u,v \in V,$ $\E[\cD(u,v)]=\sd_\bx(u,v)$. By linearity of expectation and Markov inequality, with probability $1/65$ we get,
$$ \sum_{u,v\in V: u\notin A \text{ or } v\notin A} \cD(u,v) \leq (1-1/64) n^2.$$
Thus, $\sum_{u,v\in A} \cD(u,v)\geq n^2/64$. By \eqref{eq:expdistance}, union bound and Markov inequality we get the following: there is a universal constants $c'_4 := 100 \cdot c_4$, such that with a constant probability 
we get a metric $\d(\cdot,\cdot)$ satisfying all the following conditions:
\begin{eqnarray}
 \sum_{(u,v) \in E} \d(u,v) &\lesssim& \sum_{(u,v)\in E} \sd_\bx(u,v) \nonumber \\
 \sum_{u,v \in A} \d(u,v) &\geq& n^2/64. \label{eq:avgsadistance} \\
 \sum_{u\in A} \d(u,c(u)) &\leq& c'_4 \cdot |A| \label{eq:avgballdistance} 
 \end{eqnarray}

Furthermore, $d(.,.)$ is an integral cut on $C$, i.e., for all three vertices $u,v,w\in C$ there are two whose distance is zero w.r.t.  $d(.,.)$. Define the Frechet embedding
\[ f(v):= d(a_1,v) \]
Then since $d(.,.)$ is a metric, we have
\[ \sum_{(u,v)\in E} |f(u)-f(v)| \leq \sum_{(u,v)\in E}  d(u,v) \lesssim \sum_{(u,v)\in E}  \sd_\bx(u,v) \]
We show that
$ \sum_{u,v\in A} |f(u)-f(v)| \geq c'_4 n^2 $
for a sufficiently small $c'_4>0$, and this completes the proof using \autoref{fact:frechet}. 

\begin{claim}
For every pair of vertices $u,v\in A$,
$$ |d(a_1,u) - d(a_1,v)| \geq d(u,v) - 2d(u,c(u)) - 2(v,c(v)).$$
\end{claim}
\begin{proof}
We consider two cases. i) $d(c(u),c(v))=0$ ii) $d(c(u),a_1)=0$. 
Since $d(.,.)$ defines an integral cut on $C$ one of the two cases would happen without loss of generality.

{\bf case i) } This case follows from the fact that the LHS is always non-negative but the RHS  is less than or equal to zero. This is because $|d(a_1,u)-d(a_1,v)|\geq 0$ but by triangle inequality 
$d(u,v) \leq d(u,c(u)) + d(v,c(v))$.

{\bf case ii) } This case follows by the assumption that $d(.,.)$ satisfies triangle inequality.
$$ |d(a_1,u)-d(a_1,v)| = |d(c(u),u) - d(c(u),v)| \geq d(v,c(u)) - d(u,c(u)) \geq d(u,v) - 2d(u,c(u)).$$
\end{proof}
Therefore,
$$ \sum_{u,v\in A} |f(u) - f(v)| \geq \sum_{u,v\in A} d(u,v) - 2|A| \sum_{u\in A} d(u,c(u)) \geq n^2/64 - 2c'_4\cdot n^2 \geq c'_4 n^2.$$
where the second inequality follows by \eqref{eq:avgsadistance} and \eqref{eq:avgballdistance},
and the last inequality follows by letting $c'_4 < 1/256$, and $c_4 < 1/25600$.
\end{proofof}

\bibliographystyle{alpha}
\bibliography{references,trees}

\end{document}